\documentclass[Royal]{sagej}

\usepackage{moreverb,url}
\usepackage[colorlinks,bookmarksopen,bookmarksnumbered,citecolor=blue,urlcolor=blue]{hyperref}
\usepackage{tikz-opm}
\usepackage{amsmath, amsthm, amscd, amsfonts,graphicx}
\newtheorem{theorem}{Theorem}[]

\newtheorem{proposition}[theorem]{Proposition}

\theoremstyle{definition}
\newtheorem{definition}[theorem]{Definition}

\newcommand\BibTeX{{\rmfamily B\kern-.05em \textsc{i\kern-.025em b}\kern-.08em
T\kern-.1667em\lower.7ex\hbox{E}\kern-.125emX}}

\begin{document}


\title{The dynamic framework of decision-making}

\author{Gholamreza Askari\affilnum{1}, Madjid Eshaghi Gordji\affilnum{1} and Ali Zarei\affilnum{2}}

\affiliation{\affilnum{1}Department of Mathematics, Semnan University P.O. Box 35195-363, Semnan, Iran\\
\affilnum{2}Department of Civil Engineering, Semnan University P.O. Box 35195-363, Semnan, Iran}


\email{g.askari@semnan.ac.ir; meshaghi@semnan.ac.ir; alizarei@semnan.ac.ir}

\begin{abstract}
This work explores dynamics existing in interactions between players. The dynamic system of games is a new attitude to modeling in which an event is modeled using several games. The model allows us to analyze the interplay capabilities and the feasibility objectives of each player after a conflict with other players objectives and capabilities. As an application we model relations between Soviet Union and America after World War II to October 1962, by using the dynamic system of games. The dynamic system of games as an important insight clearly has significant implications for modeling strategic interactions in which player pursue goals for increasing their personal interests. In addition, we introduce a new game in which there is a dilemma which this dilemma occurs in most societies. We investigate depends on the claim that each player in this dilemma is hyper-rational. In concept of hyper-rational, the player thinks about profit or loss of other actors in addition to his personal profit or loss and then will choose an action which is desirable to him. In this dilemma, a weak trust has been created between players, but it is fragile.
\end{abstract}


\maketitle

\section{Introduction}
Game theory provide the theoretical underpinnings of analytical techniques and their application in social research.  Indeed more fields of science have benefited from game theoretic models. Dynamic games provide a framework for modeling the behavior of players in situations where there are dynamic strategic interactions. Long provides a survey of models of dynamic games in industrial organization \cite{Van}.  At the \textit{Theory of Games and Economic Behavior}, write: We repeat most emphatically that our theory is thoroughly static. A dynamic theory would unquestionably be more complete and therefore preferable \cite{Von}. Eshaghi and Askari, recently introduced a new method of modeling in game theory, named the dynamic system of strategic games \cite{Eshaghi}. In this modeling, new properties of games such as game-maker game, strategy-maker game and the pair of rational actions are presented and with the help of these properties, the dynamics of players behavior are studied. According to this feature, strategic games were divided into two classes, strategy maker games, and games that aren't strategy maker. Also, Strategy maker games itself are of two groups and games that aren't strategy maker itself are two groups. Here we present required concepts and terms, some advantages of this modeling and  prove two propositions. Exactly similar to everyday life or the process of negotiations, in the dynamic system of games after the choice of strategy or rational action pairs by the players, by entering into the new games the productive game is out of reach. Therefore, the producer games are out of reach of the players and only as a history of the system that can influence the selection of strategies and games in the future.

As an application of the dynamic system of strategic games, we seek modeling of conflict between Soviet Union and America after World War II to October 28, 1962. The Cold War was a state of geopolitical tension after World War II between powers in the Eastern Bloc and powers in the Western Bloc. Prior to 1957, Cuba had for many years been under the complete economic and political control of the United States. In late 1956 a  revolution  to  overthrow  the  Batista  regime was initiated  by  Fidel Castro, and in  1959 he managed to take control of Cuba. The situation was made more unstable by the unsuccessful Bay of Pigs and a declaration by the Soviet Union of its willingness to aid Cuba in defending itself against the US \cite{Fraser}. Countries Soviet Union and America also belong to different models of international relation and political systems. For modeling of Cuban missile crisis, we divide the period into 5 parts and in each period we examine static games with perfect information which had occurred. In return for a public US pledge not to invade Cuba and a private assurance that the US controlled missiles in Turkey would eventually be dismantled, the Soviets agreed to withdraw the missiles \cite{Zagare}.

Theory of Game have introduced the variety of games to express the circumstances of the event \cite{Brams}. The game theorists present the variety of models for interpreting the interaction conditions that players face, which they suggest solutions to resolve these conflicts( see \cite{Yeung, Ungureanu}). In the following, we introduced Rostam's Dilemma. Rostam's Dilemma is a symmetric $2\times 2$ game that examines difficult conditions of decision making between players. Difference between this dilemma and some of other dilemmas is that in the Rostam's Dilemma, cooperation is weak dominant over non-cooperation. As an application of this Dilemma, social dilemma in the way of a wife and husband is modeled.


\section{Dynamic system of games}\label{S.2}

Maybe an event can't be modeled completely through one game but there is more chance with several games \cite{Eshaghi}. The dynamic system of strategic games is a dynamic model of $2\times 2$ games. The advantages of this modeling include:
\begin{itemize}
\item  Considering the impact of games on each other
\item  Using several games to model an event
\item  Dynamic interaction between players
\item  Performance of players' rationality within the system
\item  Selecting an agreed mechanism between players to achieve cooperation
\item  Finding solutions to conflicts
\end{itemize}

In dynamics of the decision-making process in traditional game theory by modeling the game in its extensive form, in each node one player was able to make a decision, but in dynamic system of games in each node, two or more players are able to make decisions. The inability to model the dynamical element of game play in static game theory, and the extent to which dynamic system of games naturally incorporates dynamical considerations, reveals an important virtue of dynamic system of games.

Here, we consider strategic $2\times 2$ games with perfect information. If a game produces other games, it is called \textit{game-maker game}. In general, if the games $g_1, g_2, ...,g_n$ generate games $g^{'}_1, g^{'}_2, ...,g^{'}_m$, then $g_i$- and $g^{'}_i$-s are called \textit{producer} and \textit{produced}, respectively. We call the form of displaying game-maker games as \textit{dynamic system of strategic games}.

If a game create one or more strategies is will be called \textit{strategy-maker game}. Each \textit{strategy} has at least two pairs of actions. Each pair of actions includes players' payoffs. The produced strategies can be dominant strategy, dominated strategy, weakly dominant strategy and weakly dominated strategy. Therefore, each dominant action of a player can be called \textit{dominant strategy} of a player. If a game doesn't generate any strategy, the game isn't strategy maker.

In a strategic game with ordinal preferences, player $i$s action  $a^{''}_{i}$ strictly dominates her action  $a^{'}_{i}$ if
\begin{equation*}
u_i(a^{''}_{i}, a_{-i})> u_i(a^{'}_{i}, a_{-i}) \;\  for \;\ every \;\ a_{-i}\in A_{-i},
\end{equation*}
where  $u_i$ is a payoff function that represents player i’s preferences \cite{Webb}. If for player $i$ the action $a^{''}_{i}$ is preferred to action $a^{'}_{i}$ per every choice of action of other players, it is called \textit{strictly dominant strategy} and is shown by $S^{j}_{i}$ where ${ _{k}S^{j}_{i}}$ shows $j$-th strategy of $i$-th player from $k$-th game.

In a strategic game with ordinal preferences, player $i$’s action $a^{''}_{i}$ weakly dominates her action $a^{'}_{i}$ if
\begin{equation*}
u_i(a^{''}_{i}, a_{-i})\geq u_i(a^{'}_{i}, a_{-i}) \;\  for\; every \;\  a_{-i}\in A_{-i}
\end{equation*}
and
\begin{equation*}
u_i(a^{''}_{i}, a_{-i})> u_i(a^{'}_{i}, a_{-i}) \;\ for \;\ some \;\ a_{-i}\in A_{-i},
\end{equation*}
where  $u_i$ is a payoff function that represents player i’s preferences  \cite{Webb}. If for player $i$ the action $a^{''}_{i}$ is preferred over action $a^{'}_{i}$ for each action choice of other players, it is called \textit{weakly dominant strategy} and will be represented by $S^{j}_{i}$.

If a game with $n$ players is strategy maker for $k$ players ($1\leq k \leq n$) it is called \textit{strategy maker game of order} $(n, k)$. If a game with $n$ players isn't strategy maker it is called \textit{strategy maker game of order} $(n, 0)$. In other word, we can consider a strategy maker game of order $(n, 0)$ as a game which is not strategy maker. In the following, we introduce pair of rational actions. Players’ preferences on pairs of rational actions are based on payoffs that they obtain.
\begin{definition}(\textit{Pair of rational actions}) A pair of actions is called \textit{rational} if at least hold true in one of the following conditions:
\begin{itemize}
\item  would be Nash equilibrium;
\item  pair of actions, would be Pareto dominant for both players over other pairs of actions;
\item  for each game that is strategy maker of order $(2, 1)$, pairs of rational actions for one player is responses to dominant strategy or weakly dominant strategy produced for other player.
\end{itemize}
\end{definition}

In a strategy maker game of order $(2, 2)$ where both players have strictly dominant strategy and the game hasn't Pareto action pairs over Nash equilibrium, the Nash equilibrium of game is the only rational actions pair.  In Fig.\ref{fig:A},  Stag Hunt game $g_1$ is strategy maker of order $(2,0)$. This game has Nash equilibriums $({_{1}C},{_{1}C})$ and $({_{1}D},{_{1}D})$. Also in this game, $({_{1}C},{_{1}C})$  is the dominant Pareto compared to pair of action $({_{1}D},{_{1}D})$. Therefore, set of rational actions pair of players includes $({_{1}C},{_{1}C})$  and $({_{1}D},{_{1}D})$.

\begin{figure}
\centering
\begin{tikzpicture}
\node [ opmobject] (4) {\begin{tabular}{c|c|c|}
\multicolumn{1}{c}{$g_1$} & \multicolumn{1}{c}{} & \multicolumn{1}{c}{} \\[-2.5mm]
  \multicolumn{1}{c}{} & \multicolumn{1}{c}{${_{1}C}$ } & \multicolumn{1}{c}{${_{1}D}$} \\ \cline{2-3}
  ${_{1}C}$ & 4,4 & 1,3 \\ \cline{2-3}
${_{1}D}$ & 3,1 & 2,2 \\ \cline{2-3}
\end{tabular}};
\end{tikzpicture}
\caption{The Stag Hunt game $g_1$ is a strategy maker game of order $(2, 0)$. Each player has two action ${_{k}C}$ or action ${_{k}D}$ where $k$-th shows number of game.}
\label{fig:A}
\end{figure}

 A tools that can display dynamic system of strategic games is games graph. Within each node of a graph, there is a strategic game in which players can make decisions. Each node  of this graph can be generator of the next game through the two following methods and be connected to it:
\begin{enumerate}
\item  strategies
\item  pair of rational actions
\end{enumerate}

Moving to the next node by the made strategy is always preference of one of players but continuing game with rational action pair may be preference of one player or both of them. Graph nodes are considered as \textit{initial node}, \textit{move node} and \textit{final node}. Nodes that players desire to continue the game are called \textit{move node}. \textit{Final node} is a node that players have no desire to continue the games (system). Every dynamic system of games includes players set, strategies set, set of rational actions pairs, system history and node and systemic preferences of players.

A graph $\mathcal{G}$ is given by $(G, E)$, where $G=\{g_1, g_2, ..., g_n\}$ is a finite set of
nodes or vertices that each node of this graph is a strategic game and $E=\{g_ig_j, g_sg_r,..., g_kg_l\}$ is a set of pairs of vertices (or 2-subsets of $E$) called branches or edges which indicates which nodes are produced strategies or pair of rational actions. Set of all strategies produced by $k$-th game is represented by ${_{k}\mathcal{S}}={_{k}\mathcal{S}_1}\cup {_{k}\mathcal{S}_2} \cup \emptyset$. Set of all pairs of actions players' in $k$-th game is shown with ${_{k}\mathcal{A}}={_{k}A_1}\times {_{k}A_2}$. Set of all pairs of rational actions  for player $i$ is shown with ${_{k}\mathcal{A}^{'}_{i}}$ that is a subset of  ${_{k}\mathcal{A}}$, for all $k\in\{1,2,...,n\}$.

Let $\mathcal{A}={_{1}\mathcal{A}} \cup {_{2}\mathcal{A}} \cup ... \cup {_{n}\mathcal{A}}\cup \emptyset$ and $\mathcal{S}={_{1}\mathcal{S}}\cup {_{2}\mathcal{S}} \cup...\cup {_{n}\mathcal{S}}$ be two set. The set valued functions, rational actions pair $\phi^{'}_{i}: G\to \mathcal{A}$ and strategy maker $\phi_i: G\to \mathcal{S}$ for players $i$'s are defined as follows:
\begin{equation*}
\phi^{'}_{i}(g_k)={_{k}\mathcal{A}^{'}_{i}}=
\left\{\begin{array}{rl}
\{({_{k}a_i}, {_{k}a_{-i}})_i | ({_{k}a_i}, {_{k}a_{-i}})_i\in {_{k}\mathcal{A}} \ \ \ \  if~g_k~has~pair~of~ \  \ \ \ \ \  \ \ \ \  \\  rational~actions \ \ \ \ \ \ \ \ \ \ \ \  \\
\emptyset \ \ \ \ \ \ \ \ \ \ \ \ \  ~if~g_k~has'nt~pair~of \ \ \ \ \ \ \ \ \ \ \  \\
 \ \ \ \ \ \ rational~actions, \ \ \ \ \ \ \ \ \ \ \ \
\end{array}\right.
\end{equation*}

\begin{equation*}
\phi_{i}(g_k)={_{k}\mathcal{S}_i}=
\left\{\begin{array}{rl}
\{{_{k}S^{j}_{i}}| {_{k}S^{j}_{i}}\in {_{k}\mathcal{S}}\} \ \ \ \ \ \ if~g_k~is~strategy~maker~for~player~i \\
\emptyset \ \ \ \ \ \ \ \ \ \ \ \ \ \ \ \ \ \ \ ~if~g_k~is'nt~strategy~maker~for~player~i,
\end{array}\right.
\end{equation*}
for all $i\in N$ and $j,k\in I=\{1,2,...,n\}$, where $g_k$ shows $k$-th game, $({_{k}a_i}, {_{k}a_{-i}})_i$ shows rational actions pair of $i$-th player from $k$-th game and  ${ _{k}S^{j}_{i}}$ shows $j$-th strategy of $i$-th player from $k$-th game. Every move of system as a member of set $M$ is as follows:
\begin{align*}
&M:=\Big\{m^{j}_{k}| \;\ m^{j}_{k}={_{k}S^{j}_{i}}\;\  or \;\  m^{j}_{k}=({_{k}a_i}, {_{k}a_{-i}})_i \ \ or \\
 & \ \ \ \ \ \ \ \ \ \ \ \ \ \ \   m^{j}_{k}=({_{k}a_i}, {_{k}a_{-i}})_{i,j} \ \ \forall \;\ {_{k}S^{j}_{i}}\in {_{k}\mathcal{S}_i}, \\
 & \ \ \ \ \ \ \ \ \ \ \ \ \ \ \  ({_{k}a_i}, {_{k}a_{-i}})_{i} \in {_{k}\mathcal{A}_i}, \ \  ({_{k}a_i}, {_{k}a_{-i}})_{j} \in {_{k}\mathcal{A}_j} \Big\},
\end{align*}
where $m^{j}_{k}$ shows $j$-th move  of $k$-th game and $({_{k}a_i}, {_{k}a_{-i}})_{i,j}$ shows the pair of rational action selected by players $i$ and $j$ of $k$-th game. Players’ move function $\varphi_{i}: M \to G^2$ and $\varphi_{i,j}: M \to G^2 \cup \emptyset$ with $ \varphi_{i,j}({_{k}S^{j}_{i}})=\emptyset$ is defined as following:
\begin{equation*}
\varphi_{i}(m^{j}_{k})=
\left\{\begin{array}{rl}
 (g_k, g_p)=g_kg_p & ~~if~~ m^{j}_{k}={_{k}S^{j}_{i}} \\
 (g_k, g_q)=g_kg_q & ~~if~~ m^{j}_{k}=({_{k}a_i}, {_{k}a_{-i}})_i,
\end{array}\right.
\end{equation*}
 \begin{equation*}
\varphi_{i,j}(m^{j}_{k})=
\left\{\begin{array}{rl}
\emptyset \ \ \ \ & ~~if ~~~~ m^{j}_{k}={_{k}S^{j}_{i}}\\
(g_k, g_s)=g_kg_s & ~~if ~~~~ m^{j}_{k}=({_{k}a_i}, {_{k}a_{-i}})_{i,j}.
\end{array}\right.
\end{equation*}

The above function shows by what move two game nodes have been connected to each other by one or both players. Consequently, it can be said that in move $m^{j}_{k}={_{k}S^{j}_{i}}$,  nodes $g_k$ and $g_p$ have been connected through the strategy selected by player $i$ to each other. In move $m^{j}_{k}=({_{k}a_i}, {_{k}a_{-i}})_i$ the nodes $g_k$ and $g_q$ have been connected by pair of rational action selected by player $i$ to each other. In move $m^{j}_{k}=({_{k}a_i}, {_{k}a_{-i}})_{i,j}$ the nodes $g_k$ and $g_s$ have been connected through a pair of rational actions selected by players $i$ and $j$ to each other.

Consider that $H$ is a set including all series (finite and infinite) that hold true in the following conditions:
\begin{enumerate}
\item $\emptyset$ is member of $H$.
\item Sequence $\big\{m^{j}_{i},\{g_k, m^{j}_{k}\}\big\}_{i,j,k\in I }$ for all $i, j, k\in \{1,2,...,n\}$, is a member of $H$. Each member of $H$ is called a history and is represented by $h$.
\item History $h=\big\{m^{j}_{i},\{g_k, m^{j}_{k}\}\big\}_{i,j,k\in I }$ is called final history if it is infinite or there isn't $g_{k+1}$ that is a member of h.
\end{enumerate}

The set $H$ is called \textit{system history}. Preferences of each node of a games system that are exactly the same preferences on the pairs of a strategic game actions are called \textit{node preferences} or \textit{tactical preferences}. Preferences on strategies set or set of rational actions pair of a game is called \textit{systemic preferences} or \textit{strategic preferences}.

 \begin{definition} (\textit{Dynamic system of strategic games}) A dynamic system of strategic games with perfect information including:
\begin{itemize}
 \item  a set of players
\item for each player, a set of strategies
\item for each player, a set of rational actions pair
\item system history
\item node preferences (tactical preferences) on set of all actions pairs
\item systemic preferences (strategic preferences) on strategies or pairs of rational actions.
\end{itemize}
\end{definition}

This modeling allows players to design games to earn most benefit during the negotiation and based on their bargaining power. In the produced games, players and strategies may change depending on conditions and also players and strategies may increase or decrease. This system will be provided practical freedom in which players can analyze their interaction conditions with others in the future and decide what strategy they choose and at what level has interaction. Also, a player may choose an action that has less outcome for him, in game conditions, but, obtain more benefit in future according to strategic preferences and conversely may he prefer present over future or a player may prefer tactical preferences on strategic preferences. From this modeling, the following propositions can be deduced.

\begin{proposition}\label{P.1}
 Each strategy-maker game $G$, has at least one pair of rational actions.
\end{proposition}
\begin{proof}
 Suppose $G$ is a strategy-maker game. Then game $G$, either is strategy-maker game of order $(2, 1)$ or is strategy-maker game of order $(2, 2)$. If game $G$, is strategy-maker game of order $(2, 1)$, then pairs of rational actions for one player is responses to strictly dominant strategy or weakly dominant strategy produced for other player. If game $G$, is strategy-maker game of order $(2, 2)$, both players either have strictly dominant strategy and game hasn't Pareto action pairs over Nash equilibrium, then the Nash equilibrium of game is the only rational actions pair or have weakly dominant strategy, then in this case game has more than one rational actions pair.
\end{proof}

Contrary to the above proposition, do not be right. That is, there are games that has pair of rational actions, but are not necessarily a strategy-maker.

\begin{proposition}\label{P.2}
 Each game $G$ that does not have a pair of rational actions is not strategy-maker.
\end{proposition}
\begin{proof}
The proof is by contradiction. First suppose $G$ is a strategy-maker game. Then by proposition \ref{P.1}, game $G$ has at least one pair of rational actions, which is inconsistent with main assumption of proposition \ref{P.2}.
\end{proof}
Contrary to the above proposition, do not be right. That is, there are games that are not strategy-maker, but has a pair of rational actions. In the following, we propose a different model from coincident attendance Soviet Union and America in Cuban Missile Crisis.


 \section{Cuban Missile Crisis}
 Soviet Union and America relations after World War II sometimes was strained. But the most important and dangerous crisis in Soviet Union and America relations, that put the world in the edge of a destructive war, was on establishment of Soviet Missile bases in Cuba. Soviet Union wanted to establish a massive arsenal in this region of Cuba in which deploy great number of heavy weapons to become the largest military station of this country in Caribbean. But Khrushchev’s trick was ultimately revealed. In October 1962 an American spy plane discovered the site of missile launch in the west of Havana that confirmed establishment of medium-range missiles in Cuba. To confront with danger that threat America's security, America's then president also constituted a war council and for one week, a serious and secret discussion was helpful in White House. In this Council, different suggestions including an all-out military attack to Cuba and military occupation of the country or bombarding missile launch platform in Cuba were announced. In continuance, Kenedi, in 22 October 1962 warned American people that Soviet government has established missile bases in Cuba and deploy of these missiles in a distance of one hundred mile from America beaches is a threat for the country's security. Kennedy at the same time with this warning announced that Soviet government must remove its missile bases from Cuba. Following this, he issued the order of Cuba naval blocked to prevent transfer of Soviet's new missile equipments to Cuba \cite{Kennedy}. This issue became a world crisis for one week.

 Now, using dynamic system of games, we model relations between Soviet Union and America after World War II to October 28, 1962. To this end, we divide the period into 5 parts and in each period we examine static games with perfect information which had occur. The first period is after World War II to October 13, 1962 that is shown in the form of game $g_1$. Second period from October 14 to 22 includes games $g_2$ and $g_3$. Third period is from October 23 to 25 that includes games $g_4$, $g_5$, $g_6$ and $g_7$. Fourth period is from October 26 to 27 that includes games $g_8$ and $g_9$. And finally, the day October 28 that includes game $g_{10}$. For more information  about Cuban Missile Crisis, refer to references (see \cite{Blaschke, Crall, May, Polletta}).

 Relations of Soviet Union and America after World War II and during 1950s were entered in Deadlock game and they hadn't much desire to cooperation with each other. As a result, both countries had selected non-cooperation. Deadlock game $g_1$ has been shown in Fig. \ref{fig:Cuba missile crisis}. Soviets is assumed as row player (player $1$) and America is assumed column player (player $2$). The players' set of actions include cooperation ${_{1}C}$ and defect  ${_{1}D}$. Players' preferences in this node are the same order preferences of strategic game $g_1$. The game Nash equilibrium is $({_{1}D},{_{1}D})$. In $g_1$, dominant strategy $_{1}S^{1}_{1}$ is defect  and dominated strategy $_{1}S^{2}_{1}$ is cooperation for player $1$. Also this game is producer of dominant strategy of defect $_{1}S^{1}_{2}$ and dominated strategy of cooperation $_{1}S^{2}_{2}$ for player $2$. In other words, game $g_1$ is strategy maker of order $(2, 2)$. The only pair of rational actions for players is $({_{1}D},{_{1}D})_{1,2}$. Based on players being rational, they select dominant strategy $_{1}S^{1}_{1}$ and dominant strategy $_{1}S^{1}_{2}$ for the game continuation. According to its dominant strategy, Soviet, was convinced in July 1962 to implement its Atomic Missile Establishment plan in Cuba that in order to confront with America's growing expansion. This strategy caused Soviet to hostage a part of Havana in Cuba (America's backyard). But Khrushchev’s trick was ultimately revealed. In October 1962 an American spy plane discovered the missile launch sit in the west of Havana that confirmed establishment of medium-range missiles in Cuba.

Strategy $_{1}S^{1}_{1}$ ends to Hostage game $g_2$. Soviet wanted to surprise Americans. In this game, Soviet has two actions: either cooperates ${_{2}C}$ through removing missile or defect ${_{2}D} $ and maintain missiles. America could reveal the issue ${_{2}C}$ or keep it secret ${_{2}D} $. The game Nash equilibrium is $({_{2}D},{_{2}C})$. This game is producer of dominant strategy of defect $_{2}S^{1}_{1}$ and dominated strategy of cooperation $_{2}S^{2}_{1}$ for player $1$ and also dominant strategy $_{2}S^{1}_{2}$ for player $2$ is reveal the issue and dominated strategy $_{2}S^{2}_{2}$ is keep it secret. The only pair of rational actions for both players is $({_{2}D},{_{2}C})_{1,2}$.

Strategy  $_{1}S^{1}_{2}$ ends to Self-Serving game $g_3$. In $g_3$, America has two actions: either it does nothing and makes cooperation ${_{3}C}$ or makes defect ${_{3}D}$ through blockade of Cuba. Also, Soviet also has two actions: either it cooperates ${_{3}C}$ to remove missile through diplomatic canal or makes defect ${_{3}D}$ and maintain missiles. The game Nash equilibrium is $({_{3}C},{_{3}D})$. $g_3$ is strategy maker of order $(2, 1)$. This game is producer of dominant strategy of defect $_{3}S^{1}_{2}$ and dominated strategy of cooperation $_{3}S^{2}_{2}$ for player 2. Pairs of rational actions for players are $({_{3}C},{_{3}D})_{1,2}$ and $({_{3}D},{_{3}D})_{1}$.

Strategy $_{2}S^{1}_{1}$ ends to Chicken $g_4$. In $g_4$, Soviet has two actions: either doesn't attack to America ${_{4}C}$ or attack to America ${_{4}D}$. Also, America has two actions: either it doesn't attack to Soviet ${_{4}C}$ or attacks to Soviet for retaliation ${_{4}D}$. The game Nash equilibria are $({_{4}C},{_{4}D})$ and $({_{4}D},{_{4}C})$. $g_4$ is strategy maker of order $(2, 0)$. In other words, this game isn't strategy maker for players. The players' pairs of rational actions are $({_{4}C},{_{4}C})_{1,2}$, $({_{4}D},{_{4}C})_{1}$ and $({_{4}C},{_{4}D})_{2}$.

Strategy $_{2}S^{1}_{2}$ ends to Stag Hunt $g_5$. In $g_5$, America has two actions: either through diplomacy seeks solution ${_{5}C}$ or through diplomacy put pressure on Soviet ${_{5}D}$. Also, Soviet has two actions: either it seeks solution ${_{5}C}$ through diplomacy or thinks about crisis intensification ${_{5}D}$. The game Nash equilibria are $({_{5}C},{_{5}C})$ and $({_{5}D},{_{5}D})$. $g_5$ is  strategy maker of order $(2,0)$. The players' pair of rational actions are $({_{5}C},{_{5}C})_{1,2}$ and $({_{5}D},{_{5}D})_{1,2}$.

Pair of rational actions $({_{3}C},{_{3}D})_{1}$ ends to another Stag Hunt $g_6$. In $g_6$, Soviet has two actions: either it doesn't break blockade ${_{6}C}$ or break the blockade ${_{6}D}$. Also, America has two actions: either it doesn't conflict ${_{6}C}$ with Soviet ships or conflict with Soviet ships ${_{6}D}$. The game Nash equilibria are $({_{6}C},{_{6}C})$ and $({_{6}D},{_{6}D})$. $g_6$ is  strategy maker of order $(2,0)$. The players' pair of rational actions are $({_{6}C},{_{6}C})_{1,2}$ and $({_{6}D},{_{6}D})_{1,2}$.

Strategy $_{3}S^{1}_{2}$ ends to Chicken $g_7$. In $g_7$, America has two actions: either doesn't attack to Soviet ${_{7}C}$ or attack to Soviet ${_{7}D}$ . Also, Soviet has two actions: either it doesn't attack to America ${_{7}C}$ or attacks to America in a retaliatory invasion ${_{7}D}$. The game Nash equilibria are $({_{7}C},{_{7}D})$ and $({_{7}D},{_{7}C})$. $g_7$ is strategy maker of order $(2, 0)$. The players' pairs of rational actions are $({_{7}C},{_{7}C})_{1,2}$, $({_{7}D},{_{7}C})_{1}$ and $({_{7}C},{_{7}D})_{2}$.

Based rationality of players and strategic preferences, players selecting pairs of rational actions $({_{4}C},{_{4}C})_{1,2}$, $({_{5}C},{_{5}C})^{'}_{1,2}$ and $({_{6}C},{_{6}C})_{1,2}$ end to Coordination game $g_8$. In $g_8$, Soviet has two actions: either it issues the command of ships not to move toward Cuba ${_{8}C}$ or thinks about resolving the crisis through negotiation ${_{8}D}$. Also, America has two actions: either it issues the command that America's ships not confront with Soviet ${_{8}C}$ or thinks about resolving the crisis through negotiation ${_{8}D}$. The game Nash equilibria are $({_{8}C},{_{8}C})$ and $({_{8}D},{_{8}D})$. $g_8$ is strategy maker of order $(2,0)$. The players' pairs of rational actions are $({_{8}C},{_{8}C})_{1,2}$ and $({_{8}D},{_{8}D})_{1,2}$.

Based on players' rationality and strategic preferences, players by selection of rational actions pairs $({_{5}C},{_{5}C})_{1,2}$,  $({_{6}C},{_{6}C})^{'}_{1,2}$ and  $({_{7}C},{_{7}C})_{1,2}$ ends to Rostam's Dilemma $g_9$. In $g_9$, Soviet has two actions: either it removes missile from Cuba ${_{9}C}$ or maintains missiles in Cuba ${_{9}D}$. Also, America has two actions: either it removes blockade of Cuba ${_{9}C}$ or maintains blockade of Cuba ${_{9}D}$. The game Nash equilibria are $({_{9}C},{_{9}C})$ and $({_{9}D},{_{9}D})$. $g_9$ is strategy maker of order $(2, 2)$. In $g_9$, the weak dominant strategy $_{9}S^{1}_{1}$ for Soviet is removing missiles and its weakly dominated strategy $_{9}S^{2}_{1}$ is maintaining missiles. Also, for America weakly dominant strategy $_{9}S^{1}_{2}$ is removing blockade and weakly dominated strategy  $_{9}S^{2}_{2}$ is maintaining blockade. The players' pairs of rational actions are $({_{9}C},{_{9}C})_{1,2}$,  $({_{9}D},{_{9}D})_{1,2}$, $({_{9}D},{_{9}C})_{1}$ and $({_{9}C},{_{9}D})_{2}$.

Based on players’ rationality and strategic preferences, players by selection of pairs of rational actions $({_{8}C},{_{8}C})_{1,2}$ and $({_{9}C},{_{9}C})_{1,2}$ ends to game Win-Win $g_{10}$. In this game, both players have two actions cooperation ${_{10}C}$ and defection ${_{10}D}$. The game Nash equilibrium is $({_{10}C},{_{10}C})$. In this step, the players have no appetence to continue so this completed the system. Dynamic system of game with strategic games between Soviet and America is represented by graphs in Fig. \ref{fig:Cuba missile crisis}. History of system is as follows:
\begin{align*}
&H=\Big\{\emptyset, \big\{g_1, {_{1}S^{1}_{1}}, {_{1}S^{1}_{2}} \big\},\big\{{_{1}S^{1}_{1}},\{g_2, {_{2}S^{1}_{1}}, {_{2}S^{1}_{2}}\}\big\}, \big\{{_{1}S^{1}_{2}},\{g_3, {_{3}S^{1}_{2}}, ({_{3}C},{_{3}D})_{1}\}\big\},\\
&\;\;\;\;\;\;\;\;\;\
\big\{{_{2}S^{1}_{1}}, \{g_4, ({_{4}C},{_{4}C})_{1, 2}\}\big\}, \big\{{_{2}S^{1}_{2}},\{g_5, ({_{5}C},{_{5}D})_{1,2}, ({_{5}C},{_{5}D})^{'}_{1,2}\}\big\}, \\
&\;\;\;\;\;\;\;\;\;\
\big\{({_{3}C},{_{3}D})_{1},\{g_6, ({_{6}C},{_{6}D})_{1,2}, ({_{6}C},{_{6}D})^{'}_{1,2}\}\big\}, \big\{{_{3}S^{1}_{2}}, \{g_7, ({_{7}C},{_{7}C})_{1, 2}\}\big\},\\
&\;\;\;\;\;\;\;\;\;\
\big\{({_{4}C},{_{4}D})_{1,2}, ({_{5}C},{_{5}D})^{'}_{1,2}, ({_{6}C},{_{6}D})_{1,2}, \{g_8, ({_{8}C},{_{8}C})_{1, 2}\}\big\}, \\
 &\;\;\;\;\;\;\;\;\;\
\big\{({_{5}C},{_{5}D})_{1,2}, ({_{6}C},{_{6}D})^{'}_{1,2}, ({_{7}C},{_{7}D})_{1,2}, \{g_9, ({_{9}C},{_{9}C})_{1, 2}\}\big\}, \\
&\;\;\;\;\;\;\;\;\;\
\big\{({_{8}C},{_{8}C})_{1,2}, ({_{9}C},{_{9}D})_{1,2}, \{g_{10}\}\big\}\Big\}.
\end{align*}
\begin{figure}
\centering
\includegraphics[width=.8\linewidth]{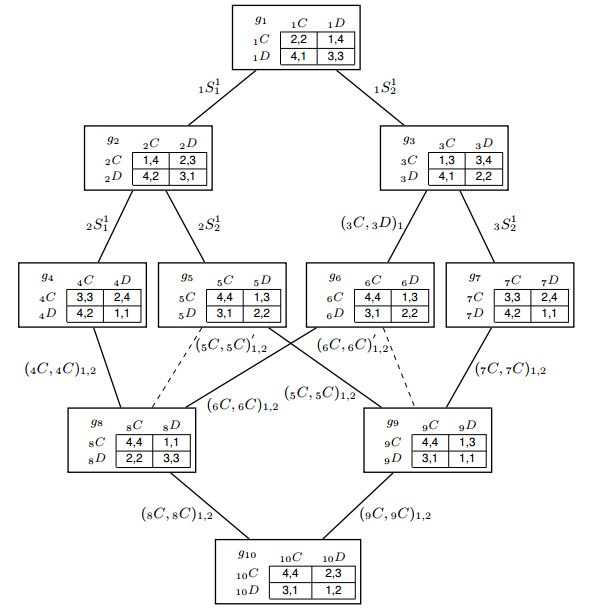}
\caption{Dynamic system of games between Soviet Union and America}
\label{fig:Cuba missile crisis}
\end{figure}

In Fig. \ref{fig:Cuba missile crisis}, if Chicken had dominant strategy even for one player, atomic war was probably occurred. This weakness of not having dominant strategy in this game causes the players to consider their strategic preferences and not to use tactical preferences inside the game. In other words, players prefer strategic preferences over tactical preferences. Therefore, it can be concluded that in Cuban missile crisis, the two countries reached to compromise.


\section{Rostam's Dilemma (Weak Trust Game)}

In section Cuban missile crisis, game $g_9$ in Fig. \ref{fig:Cuba missile crisis}, is one of the other several dilemmas that we confront in attempt to achieve cooperation. In this game, cooperation can produce the best possible overall result, but there is a non-cooperative Nash equilibrium that wants to draw us toward itself. Difference between this dilemma and some of other dilemmas is that in the Rostam's Dilemma, cooperation is weak dominant over non-cooperation. Despite Prisoner's Dilemma that players have no trust in each other, in Rostam's Dilemma, there is a weak trust between players but yet this weak trust doesn't ensure that cooperation completely. There has been a view that if players trust in each other, they will obtain better result but this game shows that despite players’ weak trust in each other, there is no definite guarantee to achieve desirable result and this trust is fragile.

Abolghasem Ferdosi Toosi is an Iranian epic poet and composer of Iran national epic Shahnameh that stated conflict between Rostam and Sohrab in the epic form \cite{Ferdowsi}. Rostam and Sohrab epic is one of the saddest events of Shahnameh. Rostam is one of the Iranian athletes that marry to Samangan king's daughter, Tahmineh. Some days after marriage, he said goodbye to Tahmineh, and happily came to Iran from there went to Zabolestan. After nine months Tahmineh gave birth to a boy and informed Rostam. He was named Sohrab. One day Sohrab went to his mother and said: who is my father? If someone asks me what I say in answer? Mother said: you are the son of robust athlete Rostam and from Sam and Zal race.

In fact there are several intertwined pieces in Rostam and Sohrab story but the main conflict is between Rostam and Sohrab. After some years Soharb with an army of Tooranian and Samanganian depart for war with Iranians. When Rostam reached to Sohrab on the battle plain said: let go from here to another side and fight. Sohrab agreed and demanded person to person war and said: you are old and not able to fight against me. Rostam said: calm down. Many demons were killed in my hand so wait to see me in fight. I don't like to fight with you and kill you. Sohrab suddenly asked: who are you and from what race? I think you are Rostam. Rostam said: no I’m not. Both went to battle field and fought to the end of day. Rostam said it is night, tomorrow we wrestle. Rostam said himself: I have a son from Tahmineh who is a youth as old as Sohrab and maybe he is himself.

By sunrise, they went again to battle field. Sohrab said with himself: the more I watch him, the more I think that he is Rostam himself and I have not to fight with him. Sohrab said to Rostam: how was last night? Come on sit down to speak with each other and don't fight. My heart is drawn to you. I asked to know your name a lot of time but no one told me your name  then don't hide your name. Rostam said: last night, we talked about wrestling you can’t deceive me. Then they wrestled and fought for a while finally Sohrab took Rostam's belt and threw him on the ground and took out dagger but Rostam said: our tradition is that one who throws a person on the ground doesn't kill him first time but in the second time kill him. Sohrab agreed because he was both brave and chivalrous. They went again to battle field and grappled, this time Rostam threw Sohrab on the ground then took out his dagger and killed him.

With help of game theory, the game between Rostam and Sohrab is modeled. Rostam has two actions; either he says his name to Sohrab $S$ or not says his name to Sohrab $NS$. Sohrab also has two actions; either he fights against Rostam $F$ or not to fight against Rostam $NF$. We call this game as Rostam's Dilemma. In Rostam's Dilemma, Sohrab is considered as row player (player 1) and Rostam is considered as column (player 2)(Fig. \ref{fig:RS}G). Because in static game, players choose their action simultaneously.
\begin{figure}
\centering
\begin{tikzpicture}
\node [ opmobject] (4){\begin{tabular}{c|c|c|}
\multicolumn{1}{c}{$G$} & \multicolumn{1}{c}{} & \multicolumn{1}{c}{} \\[-2.5mm]
  \multicolumn{1}{c}{} & \multicolumn{1}{c}{S} & \multicolumn{1}{c}{NS} \\ \cline{2-3}
   NF & 4,4 & 1,3 \\ \cline{2-3}
   F & 3,1 & 1,1 \\ \cline{2-3}
\end{tabular}};
\node [ opmobject, right=of 4, xshift=4 pt, ] () {\begin{tabular}{c|c|c|}
 \multicolumn{1}{c}{$G^{'}$} & \multicolumn{1}{c}{} & \multicolumn{1}{c}{} \\[-2.5mm]
  \multicolumn{1}{c}{} & \multicolumn{1}{c}{C} & \multicolumn{1}{c}{B} \\ \cline{2-3}
   C & 4,4 & 1,3 \\ \cline{2-3}
   B & 3,1 & 1,1 \\ \cline{2-3}
\end{tabular}};
\end{tikzpicture}
\caption{$G$ is a Rostam's Dilemma. $G^{'}$ is a game between wife and husband.}
\label{fig:RS}
\end{figure}

The above game is a symmetric game. Row player has weak dominant action $NF$. Column player has weak dominant action $S$. This game has two Nash equilibria $(NF, S)$ and $(F, NS)$. Rostam's Dilemma is strategy maker of order $(2, 2)$. Row player has weak dominant strategy $NF$ and weak dominated strategy $F$. Column player has weak dominant strategy $S$ and weak dominated strategy $NS$. Also the game has four pairs of rational actions $(NF, S)_{1,2}$, $(F, NS)_{1,2}$, $(NF, NS)_{2}$ and $(F, S)_{1}$.

Recently Eshaghi and Askari introduced a new concept of rational choice called  \emph{hyper-rational choice}. In this concept, the actor thinks about profit or loss of other actors in addition to his personal profit or loss and then will choose an action which is desirable to him \cite{Askari}. This concept explains that, based on the loss of player 2, $F$ is a strictly dominant action for player 1, and  based on the loss of player 1, $NS$ is a strictly dominant action for player 2. If interaction between players is based on loss of other player, player 1 prefers: $(F, NS)_{1} \sim^{'} (F, S)_{1} \succeq^{'} (NF, NS)_{1} \succeq^{'} (NF, S)_{1}$, and player 2 prefers: $(F, NS)_{2} \sim^{'} (NF, NS)_{2} \succeq^{'} (F, S)_{2} \succeq^{'} (NF, S)_{2}$. Therefore, for both players two pairs of actions $(NF, S)$, $(F, NS)$ are hyper-rational.

Rostam's game put a dilemma on the way of players. Yet two players have weak dominant actions and can choose that and gain more payoff. In the other words, each player’s weak dominant action in some kind build a weak trust between players and this weak trust is fragile. But the game is dilemma because players can select weak dominated action and obtain least payoff. In game $g_9$ of Fig. \ref{fig:Cuba missile crisis}, both players based on system dynamic and weak trust to each other selected weak dominant strategy. In other words, based on collective benefit thinking, both player prefers: $({_{9}C},{_{9}C}) \succeq^{'} ({_{9}D},{_{9}D}) \succeq^{'} (({_{9}D},{_{9}C}) \succeq^{'} ({_{9}C},{_{9}D})$  or  $({_{9}C},{_{9}C}) \succeq^{'} ({_{9}D},{_{9}D}) \succeq^{'} (({_{9}C},{_{9}D}) \succeq^{'} ({_{9}D},{_{9}C})$. But in the game between Rostam and Sohrab because a game occur just once, weak trust was broken and both selected weak dominated action. In other words, based on collective loss thinking, both player prefers: $(F, NS) \succeq^{'} (NF, S) \succeq^{'} (F, S) \succeq^{'} (NF, NS)$ or $(F, NS) \succeq^{'} (NF, S) \succeq^{'} (NF, NS) \succeq^{'} (F, S)$. Weak trust between Rostam and Sohrab was broken and turned into distrust to each other and son (Sohrab) was killed by father(Rostam).

Now, as an application of the Rostam's Dilemma, social dilemma in the way of a wife and husband is modeled. Consider a couple who live together for several years, The couple are disagreeing with each other about a problem. This dispute lead to having trouble between husband and wife in some other problems. If both player condonation $C$ of these disputes and resolve all problems, then each earns reward $4$. If the husband select condonation of disputes, but the wife blame husband and intensify disputes $B$, then the wife gain payoff $3$ and the husband gets $1$ and vice versa. If the wife to blame her husband and the husband to blame her wife, more intensify disputes becomes, then both gain a payoff 1. The wife is considered as row player (player 1) and husband is considered as column (player 2)(Fig. \ref{fig:RS}$G^{'}$).

This dilemma occurs in most societies, wife and husband have weakly dominant strategy $C$, but sometimes wife and husband select weakly dominated strategy $B$, which lead to intensify disputes. The game Nash equilibria are $(C,C)$ and $(B,B)$. Eshaghi and Askari introduced a new concept that Called \emph{taxonomy} \cite{Askari}. Taxonomy of hyper-preference means that if we face an actor with two choices of hyper-preferences, she will necessarily have an opinion on which she likes more. Taxonomy of player’s hyper-preferences depends on environmental condition, the kind of behavior interactive, self-evaluation system and evaluation system of other interacting persons, helps to the wife and husband that based on collective profit prefer $(C, C)$ and based on collective loss prefer $(B, B)$.


\section{Conclusion}\label{S.5}

The dynamic system of games helps us to analyze an event by dividing it into different courses and dynamically using several games. This will enable the event analyst to evaluate the decisions and strategies that the players have chosen and achieved a reasonable and acceptable result. The system also shows the impact of players' decisions on each other and the impact of games on each other in any period.

In this study, we used the dynamic system of strategic games to investigate the interaction between Soviet Union and America after World War II until October 28, 1926, that is, the end of Cuban missile crisis. To this end, we divide this time interval into five periods and in each period, we reviewed static games with complete information that has occurred. Each country, based on its forces and capabilities, sought to achieve its goals and objectives in Cuba. Therefore, it can be concluded that in Cuban missile crisis, the two countries reached to compromise.  Dynamics system of games is a combination of dynamic and static interactive situations that are moving forward. In many cases, dynamic system of strategic games can provide a mechanism to move towards cooperation between players which helps to find a solution to exit conflicts. Moreover, we introduced Rostam’s Dilemma. Rostam’s Dilemma is a symmetric $2\times 2$ game that examines difficult conditions of decision making between hyper-rational players. In this game, a weak trust has been created between players, but it is fragile. The new attitude in this article toward $2\times 2$ games properties can result in a new and different characterization compared to the topology of $2\times 2$ games. The hyper-rational choice theory suggests that hyper-rational players have considered three classes of hyper-preferences that help determine how to behave in interactive decisions. The goal of a game analysis with two hyper-rationality players is to provide insight into real-world situations that are often more complex than a game with two rational players where the choices of strategy is only based on individual preferences.






\end{document}